\documentclass[english,nodate]{elsarticle}

\usepackage{
   amsmath,
   algorithm,
   algorithmic,
   amsfonts,
   graphicx,
   subfigure,
   epsfig,
   epstopdf
   }

\newtheorem{theorem}{Theorem}

\newtheorem{lemma}[theorem]{Lemma}

\newtheorem{corollary}[theorem]{Corollary}

\newtheorem{example}[theorem]{Example}

        {\hspace*{\fill}$\Box$\par\vspace{4mm}}

\newenvironment{proof}{\noindent{\em Proof.}}%
        {\hspace*{\fill}$\Box$\par\vspace{4mm}}

\def\eps{\epsilon}

\def\eps{\varepsilon}
\def\np{${\mathcal{NP}}$}

\def\Re{{\mathbb{R}}}
\def\p{{\mathcal{P}}}
\def\pof{{\mathrm{PoF}}} 

\def\A{a} 
\def\B{b} 
\def\z{U} 
\def\opt{x^*} 
\def\bestA{\hat a}
\def\bestB{\hat b}
\def\best{\hat u} 

\begin{document}

\title{Price of Fairness for Allocating a Bounded Resource\tnoteref{titlenote}}
\tnotetext[titlenote]{
Some of the results presented in this manuscript have been introduced in \cite{bib:sagt2015} as a contribution in the Proceedings of the 8th International Symposium on Algorithmic Game Theory (SAGT 2015) in Saarbr\"ucken, Germany, Sept. 28--30, 2015.
}
 
\author{Gaia~Nicosia}
\address{Dipartimento di Ingegneria, Universit\`a  degli Studi Roma Tre, Italy,
   \texttt{nicosia@ing.uniroma3.it}}
   
\author{Andrea Pacifici}
\address{Dipartimento di Ingegneria Civile e Ingegneria Informatica, Universit\`a degli Studi di Roma ``Tor Vergata'', Italy,
   \texttt{andrea.pacifici@uniroma2.it}}

\author{Ulrich Pferschy}
\address{Department of Statistics and Operations Research,
University of  Graz, Austria,
   \texttt{pferschy@uni-graz.at}}
\begin{abstract}
In this paper we study the problem of allocating a scarce resource among several players (or agents). A central decision maker wants to maximize the total utility of all agents. However, such a solution may be unfair for one or more agents in the sense that it can be achieved through a very unbalanced allocation of the resource. On the other hand fair/balanced allocations may be far from optimal from a central point of view. So, in this paper we are interested in assessing the quality of fair solutions, i.e.\ in measuring the system efficiency loss under a fair allocation compared to the one that maximizes the sum of agents utilities. This indicator is usually called the \emph{Price of Fairness} and we study it under three different definitions of fairness, namely maximin, Kalai-Smorodinski and proportional fairness.

Our results are of two different types. We first formalize a number of properties holding for any general multi-agent problem without any special assumption on the agents utilities. Then we introduce an allocation problem, where each agent can consume the resource in given discrete quantities (items).
In this case the maximization of the total utility is given by a Subset Sum Problem. For the resulting \emph{Fair Subset Sum Problem}, in the case of two agents, we provide upper and lower bounds on the Price of Fairness as functions of an upper bound on the items size.
\end{abstract}

\maketitle

\section{Introduction}
\label{sec:intro}

Fair allocation problems arise naturally in various real-world contexts
and are the object of study in several research areas such as 
mathematics, game theory and operations research.
These problems consist in sharing resources among several self-interested parties
(players or agents)
so that each party receives his/her due share.
At the same time the resources should be utilized in an efficient way
from a central point of view.
A wide variety of fair allocation problems have been addressed in the literature depending on the resources to be shared, the fairness criteria, the preferences of the agents, and other aspects for evaluating the quality of the allocation.

In this paper we  focus on a specific discrete allocation problem, introduced briefly in \cite{bib:sagt2015},
that can be seen as a multi-agent \emph{subset sum problem}:
A common and bounded resource (representing e.g., bandwidth, budget, space, etc.) is to be shared among a set of agents each owning a number of indivisible items. The items require
a certain amount of the resource, called item weight and the problem  consists in selecting, for each agent, a subset of items so that the sum
 of all selected items weights  is not larger than a given upper bound expressing the resource capacity.
We assume that the utility function
of each agent consists of the sum of weights over all selected items of that agent.
In this context, maximizing the resource utilization is equivalent to determining the solution of a classical, i.e.\ single agent,  subset sum problem.
Since we are interested in  solutions implementing some fairness criteria, we call the addressed problem the {\em Fair Subset Sum Problem} (FSSP).

Throughout the paper, as usual with allocation problems, 
we  consider for each agent a {\em utility function}
which assigns for any feasible solution a certain utility value to that agent.
We assume that the system utility (e.g.\ the overall resource utilization in an allocation problem)
is given by the sum of utilities over all agents.
This assumption of additivity 
appears frequently in quantitative decision analysis (cf.\ e.g.~\cite{zhsh15}).
The solution is chosen by a central decision maker
while the agents play no active role in the process.
The decision maker is confronted with two objectives:
On one hand, there is the maximization of the sum of utilities over all agents.
On the other hand, such a \emph{system optimum} 
may well be highly unbalanced. For instance, it could assign all resources to one agent only and this may have severe negative effects in many application scenarios.
Thus, it would be beneficial to reach a certain degree of agents satisfaction
by implementing some criterion of {fairness}.

Clearly, the maximum utility taken only over all fair solutions 
will in general deviate from the system optimum and thus incurs 
a loss of utility for the overall system.
In this paper we want to analyze this loss of utility implied by a fair solution
from a worst-case point of view.
This should give the decision maker a guideline or quantified argument
about the {\em cost of fairness}.
A standard indicator for measuring this {system efficiency loss}
is given by the relative loss of utility of a fair solution compared to the
system optimum in a worst-case sense,
which is called \emph{Price of Fairness} ($\pof$).

The concept of fairness is not uniquely defined in the scientific literature
since it strongly depends on the specific problem setting and also on the agents perception of what a fair solution is.
In this paper we consider three types of fair solutions,
namely proportional fair, maximin and Kalai-Smorodinski solutions
(definitions are  given in Section~\ref{sec:notation}).
 Moreover, we  formalize several properties of fair solutions---some of which have been already investigated in some specific contexts---holding for any general multi-agent problem without any specific assumption on the utility sets, contrarily from most of the scientific  literature on allocation or multi-agent problems.
The most significant part of this work is devoted to completely characterizing $\pof$ for the Fair Subset Sum Problem with two agents for the three above mentioned fairness concepts.

\subsection{Related literature}
\label{sec:lit}

Caragiannis et al.~\cite{bib:ckkk2009} were the first to introduce the concept of $\pof$
in the context of fair allocation problems:
In particular, they compare the value of total agents utility
in a global optimal solution 
with the maximum total utility obtained over all fair solutions
(they make use of several notions of fairness namely, proportionality, envy-freeness and equitability).
In \cite{bib:bft2011}, Bertsimas et al.\ focus on  proportional fairness and maximin fairness and
provide a tight characterization of the Price of Fairness for a broad family of allocation problems
with compact and convex agents utility sets.

The Price of Fairness measures the inefficiency implied by fairness constraints, similarly to the utility loss implied by selfish behavior of agents and quantified by the Price of Anarchy (see, e.g.~\cite{bib:PoA}).
From a wider perspective, many authors have dealt with the problem of balancing global efficiency and fairness
in terms of defining appropriate  models or designing suitable objective functions or determining tradeoff solutions
(see for instance \cite{bib:bft2012,bib:bw2002,bib:k2009}).
A recent survey on the operations research literature that considers the tradeoff between efficiency and equity is \cite{bib:km2015}.

\smallskip
The subset sum problem considered in this paper is related to the so-called
\emph{knapsack sharing problem}  in which different agents try to fit their own
items in a common knapsack (see for instance \cite{fy04,hifi2005}).
The problem consists in determining the solution that tries to
balance the profits among the agents by
maximizing the objective of the agent with minimum profit.
As we will see, this problem is equivalent to determining a specific type of fair solution, known as maximin solution in the literature.
%
%
 Another special knapsack problem has been addressed in \cite{bib:k2009}, where
a bi-objective extension of the Linear Multiple Choice Knapsack (LMCK) Problem  is considered. The author wants to maximize the profit while minimizing the maximum
difference between the resource amounts allocated to any two agents.

Fairness concepts have been widely studied in the context
of {\em fair division} problems, see e.g.~\cite{bib:bt1996}
for a general overview, and in many other application scenarios (mostly in  telecommunications systems \cite{bib:frf2015,bib:kmt98} and, more recently, in cloud computing  \cite{bib:gzhkss11,bib:pps15}). In particular, in \cite{bib:pps15} the authors point out that resource allocation in computing systems
is one of the hottest topics of interest for both computer scientists and economists.

Fair division includes a great variety of different problems in which a set of goods
has to be divided among several agents each having its own preferences.
The goods to be divided can be
($i$) a single heterogeneous good as in the classical cake-cutting problem (see e.g.~\cite{bib:bt1996} and \cite{audo10}, which considers price of fairness 
in the line of \cite{bib:ckkk2009}),
($ii$) several divisible goods as in resource allocation problems (see e.g.~\cite{bib:pps15}),
or ($iii$) several indivisible goods (see e.g.~\cite{bib:kla10}).
The fair subset sum problem we address is strongly related to fair division. It can be seen either as a single resource allocation problem in which the resource can be only allocated in predetermined blocks/portions (the item weights) or as a special case of the indivisible goods problem in which, due to an additional capacity constraint, only a selection of the goods can be allocated.

A different but related scenario is presented in \cite{bib:budgetgames},
where a game is considered in which several agents own different tasks
each requiring certain resources.
The agents compete for the usage of the scarce resources and have to select the tasks to be allocated.

\bigskip
The paper is organized as follows. The next section provides the basic definitions,
the formal statements for the problems studied (Section \ref{sec:def_ss}) and a summary of our results (Section \ref{sec:summary}).
Some  properties which hold for any general $k$-agent problem are given in Section \ref{sec:general},
where the special case of  problems with a symmetric structure is also addressed. In Section \ref{sec:subsetsum} we consider the fair subset sum problem with two agents in two different scenarios. In particular, in Section \ref{sec:separate} we present the results concerning the case in which the two agents have two disjoint sets of items, while in Section \ref{sec:shared} the case in which the agents share a common set of items is considered.
Finally, in Section \ref{sec:conc} some conclusions are drawn.

\section{Notation and problem statement}
\label{sec:notation}

Consider a general multi-agent problem $\p$, e.g.\ some type of resource allocation problem,
in which we are given a set of $k$ agents $\{ 1, 2, \ldots , k \}$ and let $X$ be the set of all feasible solutions, e.g.\ allocations.
Each agent $j$ has a utility function $u_j: X \to \Re^+$.
If two solutions $x$ and $y$ yield the same utility for all agents,
i.e.\ $u_j(x)=u_j(y)$ for all $j$,
then we are not interested in distinguishing between them and we  consider
$x$ and $y$ as equivalent.
Note that  we do not make any assumption on the set $X$ nor on the functions $u_j$.

We define the above problem to be \emph{symmetric} and  denote it by $\p_{sym}$,
if for any solution $x\in X$
and for any permutation $\pi$ of the $k$ agents there
always exists a solution $y\in X$ such that
$u_j(x)=u_{\pi(j)}(y)$ for all $j=1, \ldots, k$.
In other words, permuting among the agents the utilities
gained from a feasible solution
in a symmetric problem always results again in a feasible solution.

\smallskip
The global or social utility $\z(x)$ of a solution $x\in X$ is
the sum of the agents utilities given by
$\z(x)= \sum_{j=1}^k u_j(x)$.
The globally optimal solution $\opt$
is called \emph{system optimum}, its value is given by
$\z^*=\z(\opt)= \max_{x\in X} \left\{\sum_{j=1}^k u_j(x)\right\}$.

In addition to the system optimum solution we  consider
{\em fair solutions}, which focus on the individual utilities obtained by each agent.
In this paper, we  use three different notions of fairness formally defined below. Other notions of fairness, such as envy-freeness or equitability, are not considered here.

 \begin{itemize}
 \item  {\em Maximin fairness}:
Based on the principle of {\em Rawlsian justice}~\cite{bib:raw71},
a solution is sought such that even the least happy agent gains as much as possible,
i.e.\ the agent obtaining the lowest utility, still receives the highest possible utility.

Formally, we are looking for a solution $x_{MM}$ maximizing $f$,
such that $u_j(x_{MM}) \geq f$ for all $j= 1, \ldots ,k $.
Equivalently, we are looking for a solution
$x_{MM}\in X$ such that
\begin{equation}\label{eq:defMM}
x_{MM} =\arg\max_{x\in X} \min_{j= 1, \ldots ,k}\left\{u_j(x)\right\}.
\end{equation}

We  only consider Pareto efficient solutions
to avoid dominated solutions with the same objective function value.
Clearly, this does not guarantee the uniqueness of solutions.

\item {\em Kalai-Smorodinski fairness} \cite{bib:ks1975}:
A drawback of maximin fairness is the fact that an agent is guaranteed a certain level
of utility, thus possibly incurring a significant loss to the other agents,
even though the agent would not be able to gain a substantial utility when acting on its own.
In the Kalai-Smorodinski fairness concept we modify the notion of maximin fairness
by maximizing the minimum relative to the best solution that an agent could obtain.

Formally, let $\best_j =\max\{u_j(y) \mid y \in X\}$
be the maximum utility value each agent $j=1,\ldots ,k$ can get over all feasible solutions.
A Kalai-Smorodinski fair solution  $x_{KS}$
minimizes $f$, such that $f \geq \frac{\best_j -u_j(x_{KS})}{\best_j}$ for all $j= 1, \ldots , k $.
Equivalently, we are looking for a solution $x_{KS}\in X$ such that
\begin{equation}\label{eq:defKS}
x_{KS} =\arg\max_{x\in X} \min_{j= 1, \ldots ,k}\left\{\frac{u_j(x)}{\best_j}\right\}.
\end{equation}
As before, we only consider Pareto efficient solutions.
Clearly, if all agents can reach the same utility,
i.e.\ $\best_j=\best_1$ for all $j$, then $x_{KS}=x_{MM}$.

\item {\em Proportional fairness} \cite{bib:kmt98}:
A solution is {\em proportional fair}, if any other solution does not give a total relative
improvement for a subset of agents which is larger than the total relative loss inflicted on the other agents.
Note that a Pareto-dominated solution can never be proportional fair.

Formally, we are looking for a solution $x_{PF}\in X$
with $u_j(x_{PF})>0$ for all $j$,
such that
for all feasible solutions $y \in X$
    \begin{equation}\label{eq:def_propfair}
    \sum_{j=1}^k\frac{u_j(y) -u_j(x_{PF})}{u_j(x_{PF})}  \leq 0
		\quad \Longleftrightarrow\quad
		    \sum_{j=1}^k\frac{u_j(y) }{u_j(x_{PF})}  \leq k
    \end{equation}
\end{itemize}
While for any instance of the problems considered in this paper
maximin and Kalai-Smorodinski fair solutions always exist,
a proportional fair solution might not
(see e.g.\ Example~\ref{ex:ssppof}).
On the other hand, as we  show in the sequel, proportional fair solutions are always unique, if they exist.
In contrast, it should be noted that for maximin fairness and also for Kalai-Smorodinski fairness schemes, 
there may exist several different fair solutions.
In the literature, these two maximin concepts are sometimes extended
to a lexicographic maximin principle
(i.e.\ among all maximin solutions, maximize the second lowest utility value, and so on)
which still does not guarantee uniqueness of solutions.
However, this will not be a relevant issue for this paper.
In fact, our restriction to Pareto efficient solutions implies
the lexicographic principle for $k=2$ agents.

It is well known that in case of convex utility sets, the proportional fair solution
is a Nash solution, i.e.\ the solution maximizing the product of agents utilities
(cf.~\cite{bib:bft2011}).
Even for the general 
utility sets treated in this paper it is shown
in Theorem~\ref{th:nash}
that if a proportional fair solution exists
then it is the one that maximizes the product of utilities.
Observe however that the opposite is, in general, not true,
since a proportional fair solution does not always exist.

\smallskip
In order to measure the loss of total utility or overall welfare of a fair solution compared to the system optimum, we study the {Price of Fairness}  as defined in \cite{bib:bft2011}:
Given an instance $I$ of our general problem, let $\z_I(x)$ be the value of a fair solution $x$ and $\z^*_I$ be the system optimum value.
The Price of Fairness, $\pof$, is defined as follows:
\begin{equation}\label{eq:defpof}
 \pof  = \sup_I  \frac{\z^*_I-\z_I(x)}{\z^*_I}
\end{equation}
Obviously, $\pof \in [0,1]$.

Whenever it is important to distinguish among the different fairness concepts,
we denote the fair solutions as $x_{MM}$, $x_{KS}$ and $x_{PF}$
corresponding to maximin, Kalai-Smorodinski and proportional fair solutions
and their associated Price of Fairness as $\pof_{MM}$, $\pof_{KS}$ and $\pof_{PF}$.

\subsection{The Fair Subset Sum Problem}
\label{sec:def_ss}

In this paper, most of the results concern a specific resource allocation problem
in which $k=2$ agents $A$ and $B$ compete for the usage of a common resource
with a given unitary capacity $c=1$. Note that a different variant of a game-theoretic setting of SSP with two agents was recently considered in~\cite{dnps13}.
Here, we  consider two scenarios:
\begin{itemize}
\item {\em Separate items.} Each agent owns a set of  $n$ items
 having nonnegative weights $a_1, a_2, \ldots,  a_n $ for agent $A$
and $b_1, b_2, \ldots,  b_n $ for agent $B$.
Each agent can only use its own items.\footnote{The number of items for each agent is irrelevant, however it is natural to assume that both have the same number of items $n$.}
\item {\em Shared items.} There is only one set of items with nonnegative weights
$w_1, w_2, \ldots, w_n$. Both agents can access these items.
\end{itemize}
In the remainder of the paper we will frequently identify an item
by its weight $a_i$, $b_i$, or $w_i$.
Every solution $x$ of the problem consists of two (possibly empty)
subsets of items  $x^A$ and $x^B$, one for each agent.
In the separate items case,
$x^A\subseteq \{a_1, a_2, \ldots,  a_n \}$ and $x^B \subseteq \{b_1, b_2, \ldots,  b_n \}$, while  in the shared items case, $x^A, x^B \subseteq \{w_1, w_2, \ldots, w_n\}$
and $x^A \cap x^B =\emptyset$.
For every solution $x=(x^A, x^B )$  we  denote the total weight of
$x^A$ (resp.\ $x^B$) by $a(x)$ (resp.\ $b(x)$).
The utility of a solution $x$ for each agent
is given simply by its total allocated weight,
i.e.\ $u_1(x)=a(x)$ and $u_2(x)=b(x)$.
The maximum utility reachable for each agent will be denoted by
$\best_1=\bestA$ resp.\ $\best_2=\bestB$.
The crucial constraint for the allocation task consists of a capacity bound
on the total weight of items given to both agents.
By scaling we can assume without loss of generality that this bound is $1$.
Thus we have that every solution $x \in X$  must fulfill:
\begin{equation}\label{eq:capacity}
\z(x)=a(x)+b(x) 
\leq 1
\end{equation}
Obviously, the computation of the system optimum $\opt$
corresponds to the solution of a classical subset sum problem (SSP) \cite{kpp04},
where a subset of items from a given ground set is sought with total weight
as large as possible, but not exceeding the given capacity $c=1$.

\begin{quote} {\sc Fair Subset Sum Problem} (FSSP): Given a set of items (shared or separate) having nonnegative weights and a fairness criterion $F$, 
find a solution $x_F=(x^A, x^B )$  such that
(\ref{eq:capacity}) is satisfied and $x_F$ is fair.
\end{quote}
In this paper we are not addressing the problem of finding fair solutions, rather in characterizing the Price of Fairness in different cases.
From a computational point of view, the \np-hardness of FSSP follows immediately
from the complexity of SSP.
However, it is easy to design pseudopolynomial dynamic programming algorithms
for computing all Pareto efficient solutions and, therefore, the solutions for the three fairness criteria.
A sketch of such a dynamic program is given in Section~\ref{sec:conc}.

In Section \ref{sec:subsetsum}, we provide several bounds on $\pof$ for FSSP.
As we will see, it is easy to provide worst case instances with $\pof=1$, corresponding to pathological  instances in which items weights are either very large (e.g.~$1$) or  very small.
Also in the area of packing problems, similar pathological instances are often used to derive worst case results.
To avoid such unrealistic settings, for many bin-packing heuristics
the worst-case ratios are also studied subject to an upper bound on
the size of the maximum item weight and expressing these ratios as a function of this parameter, see e.g.~\cite[Sec.~2.2]{hoch97}.
 So, in Section \ref{sec:subsetsum}, we explore the same direction and
study the Price of Fairness restricted to instances with
an imposed upper bound
$\alpha\leq 1$ on all item weights.

\subsection{Summary of Results}\label{sec:summary}

In Section \ref{sec:general}, we provide some basic, yet very general results
for proportional fair solutions
valid for any  $k$-agent  problem.
In particular, we show that if there exists a proportional fair solution,
then such a solution is unique (recall that two solutions having the same utilities values for each agent are considered equivalent) and maximizes the product of agents utilities.
Similar results were derived in different contexts, here we provide simple proofs holding in a more general setting.  (For the readers' convenience these proofs are reported in the Appendix.)

Moreover, we present a general upper bound on the Price of Fairness for
any proportional fair solution, namely $\pof_{PF}\leq \frac {k-1}{k}$,
and compare this bound to the results in~\cite{bib:bft2011}.
Additionally, for two agents it is possible to show that the global utility of a proportional fair solution (if it exists) is always greater or equal than that of a maximin fair solution. This is not true anymore as soon as the number of agents becomes three. We also show that when dealing with Kalai-Smorodinski fair solutions, even when there are only two agents, no dominance relations can be established with respect to other concepts of fairness.

When the problem is symmetric, we give a full characterization of proportional fair solutions
by showing that if such a fair solution exists then it also system optimal and all agents get the same utility value.

\begin{table}[hp]
\begin{center}
\caption{Summary of results for FSPP on $\pof_{MM}$ and $\pof_{KS}$ in the separate items sets case.}\label{tab:res_pof_MM}
{
\begin{tabular}{ccc} \hline
 $\alpha$ &  Lower Bound on  $\pof_{MM}$  and $\pof_{KS}$& Upper Bound on  $\pof_{MM}$ and $\pof_{KS}$\\ \hline
 $ 1 $         & $ 1$ (Ex. \ref{ex:ssppof}) & 1 \\
 $[2/3,1]$    & $2-1/\alpha$  (Ex.~\ref{ex:ssppoflarge}) &  $2-1/\alpha$ (Thm.~\ref{th:ssppof} and \ref{th:ssppofks})\\
 $[1/2, 2/3]$ & $1/2$  (Ex.~\ref{ex:pof_small_PF}) & $1/2$   (Thm.~\ref{th:ssppof} and \ref{th:ssppofks}) \\
 $(0,1/2]$    & $\frac{1}{ \lceil \frac 1 {\alpha} \rceil } $ (Ex.~\ref{ex:pof_small_PF}) &  $\alpha $  (Thm.~\ref{th:ssppof} and \ref{th:ssppofks}) \\
\hline
\end{tabular}
} 
\end{center}
\end{table}

The main body of this paper concerns the FSSP with two agents.
The corresponding results are summarized in three tables.
Tables~\ref{tab:res_pof_MM} and \ref{tab:sep_PF} concern the separate item case under the three fairness schemes
as treated in Section~\ref{sec:separate}.
We give lower and upper bounds on $\pof_{MM}$, $\pof_{PF}$, and $\pof_{KS}$
depending on an upper bound $\alpha$ on all item weights.
Finally, Table~\ref{tab:shared} refers to the separate items case, in which
maximin and Kalai-Smorodinski fair solutions coincide and a proportional fair solution
is optimal, if it exists (see Section~\ref{sec:shared}). The results reported in the tables are also illustrated in Figure~\ref{fig:poftrends}.

\begin{table}[hptb]
\begin{center}
\caption{Summary of results for FSPP on  $\pof_{PF}$ in the separate items sets case.}\label{tab:sep_PF}
{
\begin{tabular}{ccc} \hline
 $\alpha$ &  Lower Bound on  $\pof_{PF}$ & Upper Bound on  $\pof_{PF}$ \\ \hline
$[1/2, 1]$ & $1/2$ (Ex.~\ref{ex:pof_small_PF}) &  $1/2$ (Thm.~\ref{th:prop_fair_k} or Cor.~\ref{th:cor_pofPFlessMM})\\
  $(0,1/2]$ &  $\frac{1}{ \lceil \frac 1 {\alpha} \rceil } $ (Ex.~\ref{ex:pof_small_PF}) & $\alpha$  (Lemma~\ref{lem:lessthanalpha} or Cor.~\ref{th:cor_pofPFlessMM})\\
\hline
\end{tabular}
} 
\end{center}
\end{table}

\begin{table}[htpb]
\begin{center}
\caption{Summary of results for FSPP in the shared items case
(here $\pof=\pof_{MM}=\pof_{KS}$ and $\pof_{PF}(\alpha)=0$, if $x_{PF}$ exists).}\label{tab:shared}
{
\begin{tabular}{ccc} \hline
  $\alpha$ &  Lower Bound on  $\pof$ & Upper Bound on  $\pof$ \\ \hline
  $1$     & $1$    (Ex. \ref{ex:ssppof})  & 1 \\
  $[2/3, 1]$ & $2\alpha -1$ (Ex.~\ref{ex:ssppoflargeshared}) & $2\alpha -1$ (Thm.~\ref{th:ssppofshared})\\
  $[1/3,2/3]$ & $1/3$ (Ex.~\ref{ex:shared_pof_small_alpha}) & 1/3 (Thm.~\ref{th:ssppofshared})\\
$(0,1/3]$ & $\frac {1} {1+ 2 \lceil \frac 1 {2\alpha}\rceil}$ (Ex.~\ref{ex:shared_pof_small_alpha}) & $\alpha$ (Thm.~\ref{th:ssppofshared}) \\
\hline
\end{tabular}
} 
\end{center}
\end{table}

\section{General results}\label{sec:general}

Hereafter, we present some simple, yet  general results
for different fair solution concepts.
Some of them may have been stated in different application contexts.
(For instance in  \cite{kyo12} a detailed discussion on the properties of proportional fair solutions is presented.)
However, to the best of our knowledge, they have not been previously formalized for a general multi-agent problem without any assumption on the utility sets.
We start by showing that if there exists a proportional fair solution, then it is unique,
i.e.\ any two proportional fair solutions must be equivalent. This result is known in different specific contexts (e.g.\ in telecommunications systems \cite{bib:kmt98} or in convex allocation problems \cite{bib:bft2011}),  in the Appendix  we provide a simple but general proof.

\begin{theorem}\label{thm:PFequal}
If two proportional fair solutions $x_{PF}$ and $y_{PF}$  exist, then $u_j(x_{PF})=u_j(y_{PF})$ for all $j=1,\ldots, k$.
\end{theorem}
The following theorem shows that a proportional fair solution is also a Nash solution, i.e. it is a utility product maximizer (the proof can be found in the Appendix).
A similar result is well-known for convex utility sets but we
are not aware of such a statement for general 
multi-agent problems.
It is clear that, in general, a Nash  solution  is not necessarily proportional fair, since proportional fair solutions might not exist.

\begin{theorem}\label{th:nash}
If a proportional fair solution $x_{PF}$ exist, then it maximizes the product of agents utilities, i.e.\
$$\prod_{j=1}^k u_j(x_{PF}) \geq \prod_{j=1}^k u_j(x) \quad \forall x\, \in X.$$
\end{theorem}

The following result establishes an upper bound on $\pof$ holding for any multi-agent problem.

\begin{theorem}\label{th:prop_fair_k}
If a proportional fair solution $x_{PF}\in X$ exists, then   $\pof_{PF}\leq \frac {k-1}{k}$.
\end{theorem}
\begin{proof}
Let us consider a proportional fair solution $x_{PF}\in X$ 
and let $\opt \in X$ be the system optimum. 
By definition of proportional fair solution, see 
$\sum_{j=1}^k \frac{u_j(\opt )}{u_j(x_{PF})} \leq k$.
 Since $u_j(\cdot)\geq 0$, this implies $\frac{u_j(\opt )}{u_j(x_{PF})} \leq k$ for all $j$. Hence,
 $\z(\opt)= \sum_{j=1}^k u_j(\opt)\leq k \sum_{j=1}^k {u_j(x_{PF})}= k\, \z(x_{PF})$.
Therefore,
 $$\pof_{PF}= \frac{\z(\opt)-\z(x_{PF})}{\z(\opt)}\leq \frac {k-1}{k}$$
which proves the theorem.
\end{proof}
It can be shown by the following example that the result of Theorem \ref{th:prop_fair_k} is tight.

\begin{example}\label{ex:pof_pf_tight}
Consider the natural extension of FSSP with separate items to $k$ agents.
Define an instance where agent $1$ has two items with weight $1$ and $1/k$
while each of the other $k-1$ agents only has an item of weight $\eps>0$.
Clearly, there are only two Pareto efficient solutions:
The system optimum $x_1=\opt$ has $u_1(\opt)=1$ and $u_j(\opt)=0$ for $j=2,\ldots,k$.
The second Pareto efficient solution $x_2$
gives $u_1(x_2)=1/k$ and $u_j(x_2)=\eps$ for $j=2,\ldots,k$.
By plugging in $x_1$ and $x_2$ in (\ref{eq:def_propfair}) it is easy to see that $x_2$ is a proportional fair solution.
Moreover, we have
$$\pof_{PF}\geq \frac {1- (1/k+(k-1)\,\eps)}{1} \to  \frac{k-1}{k}\,.$$
\end{example}

It should be observed  that the bound of the above Theorem~\ref{th:prop_fair_k}
is not implied by \cite{bib:bft2011}, where the bound provided in their Theorem~2 for $\pof_{PF}$
in the case of unequal maximum achievable utilities is:
\begin{equation}\label{eq:BertsimasBound}
\pof_{PF} \leq \frac {k-1}{k} + F- G
\end{equation}
where $F= \frac{\min_j \{{\best_j}\}}{\sum_{j} \best_j}$
and $G= \frac{2\sqrt{k} - 1}{k} \frac{\min_j \{{\best_j} \}}{\max_j \{ {\best_j} \} } $.
Clearly, depending on the $\best_j$ values, $F-G$ can be negative (e.g., Example~\ref{ex:mm_betterthan_pf}) or positive (e.g., Example~\ref{ex:ssppof}).

\subsection{Comparison between fair solution utilities}

Hereafter, we show that in case of two agents ($k=2$), the global value of a proportional fair solution---if it exists---is not smaller than that of a minimax fair solution.

\begin{theorem}\label{th:pfmm}
In the case of $k=2$ agents, if a proportional fair solution $x_{PF}$ exists, then
$\z(x_{PF}) \geq \z(x_{MM})$.
\end{theorem}
\begin{proof}
Assume by contradiction that there exists an instance
with $u_1(x_{PF})+ u_2(x_{PF}) < u_1(x_{MM})+ u_2(x_{MM})$.
Without loss of generality  we assume $u_1(x_{PF}) \geq  u_2(x_{PF})$.
By the definition of maximin fair solution, we know that $u_2(x_{MM})> u_2(x_{PF})$.
(If $u_2(x_{MM}) = u_2(x_{PF})$ then $x_{PF}$ would be Pareto dominated by $x_{MM}$).
From Pareto efficiency of $x_{PF}$ 
it follows that $u_1(x_{MM}) < u_1(x_{PF})$.
Let $u_1(x_{MM}) := u_1(x_{PF}) - \delta$ and $u_2(x_{MM}) := u_2(x_{PF})+\eps$ for
some values $\delta, \eps >0$.

From the definition of proportional fairness we have:
$$\frac{u_1(x_{PF}) - \delta}{u_1(x_{PF})} + \frac{u_2(x_{PF})+\eps}{u_2(x_{PF})} \leq 2
\Longleftrightarrow
\frac{\eps}{u_2(x_{PF})} \leq \frac{\delta}{u_1(x_{PF})}
\Longleftrightarrow
\eps \cdot u_1(x_{PF}) \leq \delta \cdot u_2(x_{PF})
$$
Since $u_1(x_{PF}) \geq  u_2(x_{PF})$ this implies $\eps \leq \delta$.
But then we have
$u_1(x_{MM})+ u_2(x_{MM}) = u_1(x_{PF}) - \delta + u_2(x_{PF})+\eps
\leq u_1(x_{PF}) + u_2(x_{PF})$ in contradiction to the above assumption.
\end{proof}
Theorem~\ref{th:pfmm} immediately yields the following statement for the Price of Fairness.
\begin{corollary}\label{th:pofPFlessthanpofMM}
In the case of $k=2$ agents, if a proportional fair solution $x_{PF}$ exists, then $\pof_{PF}\leq \pof_{MM}$.
\end{corollary}
As soon as the number of agents increases, already for $k=3$, this property does not hold anymore, in general. This is shown by the following example.

\begin{example}\label{ex:mm_betterthan_pf}

Consider an instance of an extension of FSSP to three agents and separate items. Let $A$, $B$, and $C$ be the three agents each owning two items denoted as $a_1$, $a_2$, $ b_1 $, $ b_2$,  $ c_1   $, and $ c_2 $. Their weights are reported in the following table.
\quad{}\\
$$
\begin{array}{lcccccccc}
\hline
\mbox{item}  & a_1  & a_2  & & b_1  & b_2  & & c_1   & c_2 \tabularnewline
\hline
\mbox{weight }  & \frac 1 5  + 2\eps & \frac 1 5  + \eps  & & \frac 1 2  + 5\eps  & \frac 1 2  + \eps & & \frac 1 4  + 7\eps & \frac 1 4  + 11\eps   \tabularnewline
\hline
\end{array}
$$
It is easy to see that, for some small $\eps$ values, e.g. $\eps=0.003$,   the solution consisting  of items $a_2$, $b_2$ and $c_2$ is a proportional fair solution $x_{PF}$ and has global value $\z(x_{PF})=0.95 + 13\eps $, 
 while the solution with items  $a_1$, $b_1$ and $c_1$ is a maximin fair solution $x_{MM}$ and has global value $\z(x_{MM})=0.95 + 14\eps$.
\end{example}
The dominance relation of Theorem \ref{th:pfmm} does not extend to Kalai-Smorodinski solutions.
In particular, we  show through two examples that Kalai-Smorodinski solutions can have a social value
larger or smaller than those of the other two types of fair solutions.
The setting of the examples follows the FSSP described in Section~\ref{sec:notation} in the case in which there are only two agents and separate item sets.
Example \ref{ex:ks_not_pf} provides an instance where  $\z(x_{PF}) = \z(x_{MM})> \z(x_{KS})$, while in Example \ref{ex:ks_betterthan_pf} an instance with $\z(x_{PF}) = \z(x_{MM}) < \z(x_{KS})$ is reported.

\begin{example}\label{ex:ks_not_pf}

Consider an instance of 
FSSP with separate items and weights as in the following table, where $0<\eps'<\eps$ are small values.
\quad{}\\
$$
\begin{array}{lcccccccc}
\hline
\mbox{item}  & a_1  & a_2 & a_3  & a_4  & b_1  & b_2 & b_3  & b_4  \tabularnewline
\hline
\mbox{weight }  & 1  & \frac 1 4 +\eps'  & \frac 1 4 & \frac 1 4 & 1- 3\eps & \frac 1 4   & \frac 1 4 & \frac 1 4 - \eps \tabularnewline
\hline
\end{array}
$$
In this case, it is possible to enumerate all six Pareto efficient solutions,
whose utilities are reported below:
$$
\begin{array}{lcccccc}
\hline
\mbox{Solution } x   & x_1  & x_2 & x_3  & x_4  & x_5  & x_6  \tabularnewline
\hline
{\A(x)}  & 1  & \frac 3 4 +\eps'  & \frac 1 2 +\eps'  & \frac 1 2 & \frac 1 4 +\eps' & 0 \tabularnewline
{\B(x)}  & 0  & \frac 1 4 -\eps  & \frac 1 2 -\eps  & \frac 1 2 & \frac 3 4 -\eps & 1-3\eps \tabularnewline
\hline
\end{array}
$$
It is easy to see that $\bestA=1$, $\bestB=1-3\eps$, and, with some simple algebra, to verify that $x_{PF}=x_{MM}=x_4$, while $x_{KS}=x_3$. Hence, in this example $\z(x_{PF}) = \z(x_{MM})> \z(x_{KS})$.
\end{example}


\begin{example}\label{ex:ks_betterthan_pf}
Consider an instance of 
FSSP with separate items and weights as in the following table
with $\eps>0$.
$$
\begin{array}{lcccccc}
\hline
\mbox{item}  & a_1  & a_2 & a_3   & b_1  & b_2 & b_3    \tabularnewline
\hline
\mbox{weight }  & 1  & \frac 3 4   & \frac 1 2 +\eps & \frac 1 4 -\eps  & \frac 1 4 - 2\eps & 0 \tabularnewline
\hline
\end{array}
$$
In this case also, it is possible to enumerate all three Pareto efficient solutions,
whose utilities are reported below:
$$
\begin{array}{lccc}
\hline
\mbox{Solution } x   & x_1  & x_2 & x_3   \tabularnewline
\hline
{\A(x)}  & 1  & \frac 3 4   & \frac 1 2 +\eps    \tabularnewline
{\B(x)}  & 0  & \frac 1 4 -\eps  & \frac 1 2 -3\eps  \tabularnewline
\hline
\end{array}
$$
Clearly, $\bestA=1$ and $\bestB=\frac 1 2 - 3 \eps$, while for $0<\eps < \frac 1 {10}$ we have that  $x_{PF}=x_{MM}=x_3$ and  $x_{KS}=x_2$. So, in this example $\z(x_{PF}) = \z(x_{MM})< \z(x_{KS})$.
\end{example}

\subsection{Symmetric multi-agent problem}\label{sec:symmetry}

Consider now a general symmetric multi-agent problem  $\p_{sym}$. Recall that, in this case, all agents are
``interchangeable'' in the sense that for any solution $x\in X$
and for any permutation $\pi$ of the $k$ agents there
always exists a solution $y\in X$ such that,
$u_j(x)=u_{\pi(j)}(y)$ for all $j=1, \ldots, k$.
This concept of symmetry applies for a large number of allocation problems
and has been often studied in the literature
(see, e.g.~\cite{bib:gkw2010}).
Also, in game theory, a symmetric game is a game where the payoffs for playing a particular strategy depend only on the other strategies employed, not on who is playing them.

The following simple result presents a necessary condition for the existence of a proportional fair solution in the symmetric case.

\begin{theorem}\label{th:shared_pf_sameamount}
If a proportional fair solution $x_{PF}$ of problem $\p_{sym}$
exists, then all the agents
have the same utility values, i.e.~$u_j(x_{PF})=\frac 1 k \z(x_{PF})$ for all $j=1,\ldots, k$.
\end{theorem}
\begin{proof}
Let $x_{PF}$ be a proportional fair solution 
and assume by contradiction that there is (at least) one pair of agents, say $1$ and $2$, having different utilities, i.e., $u_1(x_{PF}) \neq u_2(x_{PF})$. By definition of $\p_{sym}$, there exists a feasible ``permuted'' solution $y$ with $u_1(y)=u_2(x_{PF})$, $u_2(y)=u_1(x_{PF})$, and unchanged utilities $u_j(y)=u_j(x_{PF})$ for all the other agents $j = 3,\ldots,k$.

Since $x_{PF}$ is a proportional fair solution and $y$ is a feasible solution, from  \eqref{eq:def_propfair} we have that:
$$
\frac{u_1(y)}{u_1(x_{PF})} + \frac{u_2(y)}{u_2(x_{PF})} + \sum_{j=3}^k \frac{u_j(y)}{u_j(x_{PF})}  \leq k
$$
which yields $\frac{u_2(x_{PF})}{u_1(x_{PF})} + \frac{u_1(x_{PF})}{u_2(x_{PF})} \leq 2$.
But this is a contradiction since, for any positive $r \neq 1$, $r + \frac 1 r > 2$.
Thus, in a proportional fair solution, no pair of agents can have different utility values and the thesis follows.
\end{proof}
Note that the  condition in Theorem \ref{th:shared_pf_sameamount} is necessary but not sufficient for a solution to be proportional fair,
see for instance Example~\ref{ex:ssppoflargeshared}.
However, it follows immediately that if a proportional fair solutions of problem $\p_{sym}$ exists,
then it must also be optimal.

\begin{corollary}\label{thm:pfsystemoptimum}
If a proportional fair solution $x_{PF}$ of problem $\p_{sym}$
exists, then it is system optimal, i.e.\ $\z(x_{PF})=\z(\opt)$
and $\pof_{PF}(\alpha)=0$, for any $\alpha \in (0,1]$.
\end{corollary}
\begin{proof}
From Theorem \ref{th:shared_pf_sameamount} we know that if
a proportional fair solution exists, then $u_j(x_{PF})=\frac 1 k \z(x_{PF})$.
Plugging in this identity into the definition of proportional fairness \eqref{eq:def_propfair}
we get:
$$ k \geq \sum_{j=1}^k\frac{u_j(\opt) }{u_j(x_{PF})}
=\sum_{j=1}^k \frac{u_j(\opt) }{\frac 1 k \z(x_{PF})}
\ \Longleftrightarrow\
\z(x_{PF}) = \sum_{j=1}^k u_j(\opt) = \z(\opt)
$$
which proves the thesis.
\end{proof}
So far, we presented some general results holding for any general multi-agent problem.
In the next section we address a specific allocation problem with $k=2$ agents.

\section{Price of Fairness for the fair subset sum problem with two agents}\label{sec:subsetsum}

In this section we focus on the Fair Subset Sum Problem (FSSP) for two agents and we  provide several bounds on the Price of Fairness.
As we discussed in Section~\ref{sec:def_ss}, to give a more comprehensive analysis, we introduce an upper bound
$\alpha\leq 1$ on the largest item weight,
i.e.\ $a_i, b_i, w_i \leq \alpha$ for all items $i$
and analyze $\pof$ as a function of $\alpha$.
Formally, we extend the definition of $\pof$ from \eqref{eq:defpof}
by taking the upper bound $\alpha$ into account:
Let $\mathcal{I}_{\alpha}$ denote the set of all instances of our FSSP
where all items weights are not larger than $\alpha$.
Given $I\in \mathcal{I}_{\alpha}$ let $\z_I(x)=a(x)+b(x)$ for a solution $x$
and $\z^*_I$ be the system optimum value for instance $I$.
Then we can define the Price of Fairness depending on $\alpha$ as follows:
\begin{equation}\label{eq:def_pof_alpha}
 \pof(\alpha)  = \sup_{I\in \mathcal{I}_{\alpha}}  \frac{\z^*_I-\z_I(x)}{\z^*_I}
\end{equation}
Obviously, $\pof=\pof(1)$.
It is also clear from the above definition that $\pof(\alpha)$ is monotonically
increasing in $\alpha$, i.e.\ if $\alpha>\alpha'$, then $\pof(\alpha) \geq \pof(\alpha')$.
Moreover, note that the value $\pof(\alpha)$  may be actually attained for an instance
$I\in \mathcal{I}_{\alpha'}$ with $\alpha' < \alpha$. Figure~\ref{fig:poftrends} illustrates the functions $\pof_{MM}(\alpha)$ and $\pof_{KS}(\alpha)$ for the separate items sets and shared items set cases.

\smallskip
The first bound on the Price of Fairness for FSSP
with $k$ agents and an upper bound $\alpha$
on the maximum item weight, is given in the following lemma. We show in the next sections that for certain $\alpha$ values this bound can be improved.

\begin{lemma}\label{lem:lessthanalpha}
The Price of Fairness for any Pareto efficient solution of the FSSP
with $k$ agents and an upper bound $\alpha \in (0,1]$
on the maximum item weight is not larger than $\alpha$,
i.e.\ $\pof(\alpha) \leq \alpha$.
\end{lemma}
\begin{proof}
We can observe that if a Pareto efficient solution $x \neq \opt$ (where $\opt$ is a system optimum) is such that
$\z(x)\leq 1 - \alpha$ then any item not included in $x$ (with weight at most $\alpha$) could be added to $x$ which thus cannot be a Pareto efficient solution.
Hence, it must be $\z(x) > 1 - \alpha$ and thus, recalling that $\z^*\leq 1$, $\pof(\alpha) \leq \frac{\z^*-(1-\alpha)}{\z^*} \leq \frac{1-(1-\alpha)}{1} \leq \alpha$, and the thesis follows.
\end{proof}
Hereafter, we discuss in detail the two scenarios introduced in Section \ref{sec:def_ss}:
The separate items case is analyzed in Section \ref{sec:separate}, while the shared items one is addressed in Section \ref{sec:shared}.

\subsection{Separate item sets}
\label{sec:separate}

Here we assume that each agent owns a separate set of items
denoted by $a_1, a_2, \ldots$ for agent $A$ and $b_1, b_2, \ldots$ for agent $B$.
To avoid trivial cases we also assume that $a_i\leq 1$ and $b_i \leq 1$
and $\sum_i a_i + \sum_i b_i >1$,
i.e.\ in every feasible solution at least one item has to remain unselected.
We start with a very simple example 
showing that in general the Price of Fairness can reach~$1$.

\begin{example}\label{ex:ssppof}
Consider an instance of the two agent FSSP with $n=2$ and items weights reported in the following table.
$$
\begin{array}{lcccc}
\hline
\mbox{item}  & a_1  & a_2  & b_1  & b_2 \tabularnewline
\hline
\mbox{weight}  & 1  & \eps  & \eps & \eps\tabularnewline
\hline
\end{array}
$$
It is easy to see that there are only two nondominated solutions, $x_1$ and $x_2$, with $\A(x_1)=1$, $\B(x_1)=0$ and $\A(x_2)=\eps$, $\B(x_2)=2\eps$. Clearly, $x_1=\opt$ is the global optimum and $\z^*=\bestA = \A^*=1$, while $x_2$ is a maximin fair solution and also a Kalai-Smorodinski solution, i.e. $x_2= x_{MM}=x_{KS}$, where  $\B(x_2)=\bestB=2\eps$.
 So, a worst possible lower bound is given by
$$\pof_{MM}=\pof_{KS} \geq  \frac{1-3\eps}{1} \to 1.$$
\end{example}
Note that in the above example, for small $\eps$ values there exist no proportionally fair solutions.

Hereafter, we introduce an upper bound $\alpha<1$ on the maximum item weight.
At first 
we give two examples providing lower bounds on $\pof(\alpha)$.
\begin{example}\label{ex:ssppoflarge}
Consider the case $\alpha \in [2/3, 1)$ and let the items weights of an instance of FSSP with two agents be reported in the following table.
$$
\begin{array}{lcccc}
\hline
\mbox{item}  & a_1  & a_2  & b_1  & b_2 \tabularnewline
\hline
\mbox{weight}  & \alpha  & 2\eps  & 1-\alpha+\eps & \eps^2\tabularnewline
\hline
\end{array}
$$
There are two nondominated solutions, namely $x_1$  with $\A(x_1)=\alpha+2\eps$ and $\B(x_1)=\eps^2$ (which is the system optimum) and $x_2$  with $\A(x_2)=2\eps$ and $\B(x_2)= 1-\alpha+\eps+\eps^2$.

It can be easily checked that $x_2$ is a maximin fair solution and also a Kalai-Smorodinski solution (i.e. $x_2= x_{MM}=x_{KS}$), while no proportionally fair solution exists for this instance.
This yields
$$\min\{\pof_{MM}(\alpha),\pof_{KS}(\alpha) \} \geq  \frac{\alpha+2\eps +\eps^2- (1-\alpha+3\eps+\eps^2)}{\alpha+2\eps+\eps^2} \to \frac{2\alpha -1}{\alpha} = 2- \frac 1 \alpha.$$
\end{example}
Note that for $\alpha \to 1$ the bound and the instance of Example~\ref{ex:ssppoflarge} tend to those of Example~\ref{ex:ssppof}.

\smallskip
The following example covers the case $\alpha \in (0, 2/3]$.
\def\k{r}
\begin{example}\label{ex:pof_small_PF}
Let agent $A$ own $\k$ items of weight $1/\k$ and $B$ own $\k$ items of weight $\eps$.
In this case there are only two Pareto efficient solutions, namely $x_1=\opt$ which is the system optimum with values $\A^*=1$ and $\B^*=0$, and $x_2$ with values  $\A(x_2)=\frac {\k-1} \k $ and $\B(x_2)=\k \eps$. It is easy to show that $x_2$ is a  fair solution in all three settings, i.e. $x_2= x_{MM}=x_{KS}=x_{PF}$.
For  $\eps \to 0$ we get:
$$\min\{\pof_{MM}(1/\k),\pof_{KS}(1/\k),\pof_{PF}(1/\k) \}\geq
\frac{1-\frac {\k-1} \k }{1}= \frac 1 \k.$$
Hence, we can state that for every $\alpha$ with
$\frac{1}{h} \leq  \alpha < \frac{1}{h-1} $, $h\geq 2$ and integer,
$$\min\{\pof_{MM}(\alpha),\pof_{KS}(\alpha),\pof_{PF}(\alpha) \} \geq
\frac{1}{h} = \frac{1}{ \lceil \frac 1 {\alpha} \rceil }.$$
\end{example}
From the bound of Example~\ref{ex:pof_small_PF} when $\k=2$,
we get $\pof_{MM}(\alpha) \geq 1/2$ for $\alpha \geq 1/2$.
Note that for $\alpha =2/3$ this matches the lower bound of Example~\ref{ex:ssppoflarge}.

 In the following theorem we  show that the bounds of Examples~\ref{ex:ssppoflarge} and \ref{ex:pof_small_PF}
 for the maximin fairness concept
are worst possible when $\alpha \geq 1/2$, i.e.\ $\pof_{MM}(\alpha)$ cannot be larger than the lower bounds provided by those examples. In Figure~\ref{fig:poftrendsseparate} the function $\pof_{MM}(\alpha)$, or the corresponding upper and lower bounds when $\alpha \leq 1/2$, are plotted
for the separate items sets case.
\begin{theorem}\label{th:ssppof}
FSSP with separate item sets and an upper bound $\alpha$ on the maximum item weight
has the following Price of Fairness for maximin fair solutions:
\begin{eqnarray}
\pof_{MM}(\alpha) & = & 2- 1/\alpha  \quad\mbox{ for } 2/3 \leq \alpha \leq 1 \label{eq:poflarge}\\
\pof_{MM}(\alpha) & =  & 1/2  \quad\mbox{ for } 1/2 \leq \alpha < 2/3\label{eq:pofmedium}\\
\frac{1}{ \lceil \frac 1 {\alpha} \rceil } \ \leq\
\pof_{MM}(\alpha) & \leq & \alpha  \quad\mbox{ for } \alpha < 1/2\label{eq:pofsmall}
\end{eqnarray}
\end{theorem}
\begin{proof}
The case  $\alpha < 1/2$ follows from Lemma~\ref{lem:lessthanalpha},
thus proving \eqref{eq:pofsmall}
with the lower bound given by Example~\ref{ex:pof_small_PF}.

\smallskip
We now consider the case $\alpha \geq 1/2$ and prove upper bounds  (\ref{eq:poflarge})
and  (\ref{eq:pofmedium}).
The corresponding matching lower bounds were given in Example~\ref{ex:ssppoflarge}
and \ref{ex:pof_small_PF} (take $r=2$).
We assume without loss of generality  that $a_1=\alpha$.
If the fair solution $x_{MM}$ includes an item with weight $\alpha$, we have $\z(x_{MM})\geq 1/2$ and thus $\pof_{MM}(\alpha) \leq 1/2$.
Hence, we assume that  $\opt$ includes $a_1$ since otherwise
neither $x_{MM}$ nor $\opt$ would include the largest item
and we could remove it from consideration\footnote{Note that this is not possible
for $x_{KS}$ where the largest item might contribute to $\bestA$.}.
%
Now we consider two cases:
\begin{itemize}

\item {\em Case}
$\B(x_{MM}) \leq 1-\alpha:$ In this case, it is feasible to include $a_1$ in $x_{MM}$ and thus
$\z(x_{MM})\geq 1/2$ and $\pof_{MM}(\alpha) \leq 1/2$.

\item {\em Case}
$\B(x_{MM}) > 1-\alpha:$
Let $\A(\opt)= a_1 + \delta$ for some residual weight $\delta\geq 0$.
We can assume that $\B(x_{MM}) < \alpha$, since otherwise we would have again $\z(x_{MM}) \geq \alpha\geq 1/2$ thus implying the thesis.
This means that there is enough capacity for $A$ to pack at least $\delta$ also in the fair solution,
i.e.\ $\A(x_{MM}) \geq \delta$.
Now we can distinguish two bounds on the fair solution.

Assume first that
\begin{equation}\label{eq:bnd1}
\z(x_{MM}) \geq \delta + 1 - \alpha.
\end{equation}
Since $\A(\opt) \geq 1/2$, it must be $\B(\opt) \leq \A(\opt)$ and thus $\B(x_{MM}) \geq \B(\opt) $,
but also $\A(x_{MM})\geq \B(\opt) $.		

Secondly, assume that
\begin{equation}\label{eq:bnd2}
\z(x_{MM}) \geq \B(\opt) + 1 -\alpha.
\end{equation}
If we combine (\ref{eq:bnd1}) and (\ref{eq:bnd2})
and define $u:= \max\{\delta, \B(\opt)\}$ and $v:= \min\{\delta, \B(\opt)\}$,
 we have the following:
\begin{equation}\label{eq:eq}
\pof_{MM}(\alpha) \leq \frac{\alpha+\delta+\B(\opt) - (u+1 -\alpha)}{\alpha+\delta+\B(\opt)}
=\frac{2\alpha -1 +v}{\alpha+\delta+\B(\opt)}.
\end{equation}
By elementary algebra it is easy to observe that showing that  $\pof(\alpha) \leq \frac{2\alpha-1}{\alpha}$ is
equivalent to showing
$$(2\alpha-1) \delta + (2 \alpha-1) \B(\opt) \geq \alpha v.$$
This last expression is true for $\alpha \geq 2/3$ by the definition of $v$.
\end{itemize}
Finally, for the case $1/2<\alpha < 2/3$ it can be easily shown that the desired upper bound
of $1/2$ is obtained from (\ref{eq:eq}).
\end{proof}
Since $\pof_{PF} (\alpha)\leq \min\{ 1/2, \pof_{MM}(\alpha)\}$ when a proportional fair solution exists (Theorem \ref{th:prop_fair_k} and Corollary \ref{th:pofPFlessthanpofMM}), we get the following result (see Example~\ref{ex:pof_small_PF}
for the tightness of $1/2$).

\begin{corollary}\label{th:cor_pofPFlessMM}
FSSP with separate item sets and an upper bound $\alpha$ on the maximum item weight
has the following Price of Fairness for proportional fair solutions:
\begin{eqnarray}
\pof_{PF} (\alpha) & = & 1/2  \quad\mbox{ for } 1/2 \leq \alpha \label{eq:pofKSmedium}\\
\pof_{PF}(\alpha) & \leq & \alpha  \quad\mbox{ for } \alpha < 1/2\label{eq:pofKSsmall}
\end{eqnarray}
\end{corollary}
We conclude this section by providing upper bounds on the Price of Fairness for Kalai-Smorodinski fair solutions. Note that these worst case bounds have the same values as those for maximin fair solutions, even though the proof is quite different. As for Theorem~\ref{th:ssppof}, Figure~\ref{fig:poftrendsseparate} illustrates the function $\pof_{KS}(\alpha)$ for the separate items sets case. Recall that it was established by Examples~\ref{ex:ks_not_pf} and~\ref{ex:ks_betterthan_pf}
that in general the utilities reached for the two fairness concepts have no dominance relations.
\begin{theorem}\label{th:ssppofks}
FSSP with separate item sets and an upper bound $\alpha$ on the maximum item weight
has the following Price of Fairness for Kalai-Smorodinski fair solutions:
\begin{eqnarray}
\pof_{KS}(\alpha) & = & 2- 1/\alpha  \quad\mbox{ for } 2/3 \leq \alpha \leq 1
\label{eq:poflargeks}\\
\pof_{KS}(\alpha) & = & 1/2  \quad\mbox{ for } 1/2 < \alpha < 2/3\label{eq:pofmediumks}\\
\frac{1}{ \lceil \frac 1 {\alpha} \rceil } \ \leq\
\pof_{KS}(\alpha) & \leq & \alpha  \quad\mbox{ for } \alpha \leq 1/2\label{eq:pofsmallks}
\end{eqnarray}
\end{theorem}

\begin{proof}
The lower bounds of \eqref{eq:poflargeks} and \eqref{eq:pofmediumks}
were given in Example~\ref{ex:ssppoflarge}
and \ref{ex:pof_small_PF} (take $r=2$).
The case  $\alpha < 1/2$ follows from Lemma~\ref{lem:lessthanalpha},
thus proving \eqref{eq:pofsmallks}
with the lower bound again given by Example~\ref{ex:pof_small_PF}.

When $\alpha > 1/2$ it is useful to partition the items into {\em small items} with weight at most $1/2$ and {\em large items} with weight greater than $1/2$.

Let us now consider the case $\alpha \geq 2/3$ and prove the upper bound (\ref{eq:poflargeks}).
By contradiction,  assume that $\pof_{KS}(\alpha) > 2- 1/\alpha$,
i.e.\
\begin{equation}\label{eq:pofksass}
\A(x_{KS}) + \B(x_{KS})< \frac{1-\alpha}{\alpha}\, (\A(\opt)+\B(\opt)) \leq \frac 1 2 \opt \leq  \frac 1 2\,.
\end{equation}
It follows that any remaining unpacked small item could be added to $x_{KS}$.
Thus, we conclude that {\em all} small items are included in $x_{KS}$.
If neither $A$ nor $B$ own a large item, the bound of $1/2$ would follow
from Lemma~\ref{lem:lessthanalpha}.
Furthermore, if $\opt$ does not contain a large item, then $x_{KS}=\opt$,
since $x_{KS}$ contains all small items.
Hence, we can assume w.l.o.g.\ that $A$ owns a large item, say $a'$,
which is contained in $\A(\opt)$,
and write $\A(\opt)=a'+ a^S > \frac 1 2$
for some weight sum $a^S$ comprising small items.

Due to \eqref{eq:pofksass} $x_{KS}$ does not contain $a'$, hence $a' + \B(x_{KS}) > 1$
because otherwise $a'$ could replace $\A(x_{KS})$.
Therefore,
\begin{equation}\label{eq:ksbboundl}
\B(x_{ks}) > 1-a' \geq 1- \alpha. 
\end{equation}
By the definition of Kalai-Smorodinski fair solutions, it must be:
\begin{equation}\label{eq:proofksl}
\min\left\{\frac{\A(x_{KS})}{\bestA}, \frac{\B(x_{KS})}{\bestB} \right\}
\geq
\min\left\{\frac{\A(\opt)}{\bestA}, \frac{\B(\opt)}{\bestB} \right\}.
\end{equation}
We can observe that
$$
\min\left\{\frac{\A(x_{KS})}{\bestA}, \frac{\B(x_{KS})}{\bestB} \right\}
\leq \frac{\A(x_{KS})}{\bestA}
< \frac{1/2}{\bestA}
< \frac{\A(\opt)}{\bestA}\,.
$$
Therefore, in the right-hand side of (\ref{eq:proofksl}) it must be
$\frac{\A(\opt)}{\bestA} > \frac{\B(\opt)}{\bestB}$.
This means that to fulfill (\ref{eq:proofksl}) we also must have
\begin{equation}\label{eq:kscase2l}
\frac{\A(x_{KS})}{\bestA} \geq \frac{\B(\opt)}{\bestB}.
\end{equation}
If also $B$ owns a large item, say $b'$, then $b'$ could replace $\B(x_{KS})$
because with assumption (\ref{eq:pofksass}) and (\ref{eq:ksbboundl}) we have:
$$\A(x_{KS}) + b' <
\left(\frac{1-\alpha}{\alpha}\, -\B(x_{KS})\right) + \alpha <
\frac{1-\alpha}{\alpha} -(1- \alpha) + \alpha =
\frac 1  \alpha +2 \alpha -2 \leq 1
$$
The last inequality holds exactly for $\alpha \in [1/2, 1]$.
Therefore, $B$ must own only small items,
which implies in turn that $\B(x_{KS})=\bestB$.

Considering the trivial bounds for the solutions the agents could obtain on their own,
namely $\bestA \geq \A(\opt) > 1/2$ and
$\bestB = \B(x_{KS}) < \frac{1-\alpha}{\alpha} \leq 1/2$, from \eqref{eq:kscase2l}, we get
\begin{equation}\label{eq:kscase22l}
\frac{\A(x_{KS})}{1/2} > \frac{\A(x_{KS})}{\bestA} \geq
\frac{\B(\opt)}{\bestB} > \frac{\B(\opt)}{1/2}
\end{equation}
which implies $\A(x_{KS}) > \B(\opt)$.

By assumption (\ref{eq:pofksass}) we have
$\A(x_{KS}) + \B(x_{KS})<  \frac{1-\alpha}{\alpha} \, (a'+a^S+\B(\opt))$.
Since $a^S < 1/2$ we know that $a^S$ together with $\B(x_{KS})$ would be a feasible solution.
Thus, it must be  $\A(x_{KS}) \geq a^S$.
Together with (\ref{eq:ksbboundl}) this means that the above assumption also implies
$$\A(x_{KS}) + (1-a')< \frac{1-\alpha}{\alpha}\, (a'+\A(x_{KS})+\B(\opt))$$
which reduces to
$$\alpha + (2 \alpha-1) \A(x_{KS}) < a' + (1-\alpha)\B(\opt).$$
But this is clearly a contradiction since $\alpha \geq a'$,
$\A(x_{KS}) > \B(\opt)$ and $(2 \alpha-1) \geq (1-\alpha)$ for $\alpha \geq 2/3$. Thus, bound \eqref{eq:poflargeks} is proven.

Since \eqref{eq:poflargeks} also means $\pof_{KS}(2/3) \leq  1/2$
and $\pof_{KS}(\alpha)$ is monotonically increasing in $\alpha$,
we immediately get the upper bound of $1/2$ also for $\alpha < 2/3$
as stated in \eqref{eq:pofmediumks}.
\end{proof}
While the bounds of Theorem~\ref{th:ssppof}, Corollary~\ref{th:cor_pofPFlessMM}
and Theorem~\ref{th:ssppofks} are tight for $\alpha \geq 1/2$,
there remains a gap for $\alpha < 1/2$ with
$\frac{1}{\lceil \frac 1 {\alpha} \rceil} \leq \pof(\alpha) \leq \alpha$.
The worst case for this interval arises for $\alpha= \frac 1 2 -\eps$
where we have
$\frac 1 3 \leq \pof(\alpha) < \frac 1 2$.
For smaller values of $\alpha$ with $\frac{1}{h} \leq  \alpha < \frac{1}{h-1}$,
$h\geq 3$ and integer, the ratio $r(\alpha)$ between upper and lower bound
on $\pof_{MM}(\alpha)$, $\pof_{PF}(\alpha)$ and $\pof_{KS}(\alpha)$ (see below)
can be bounded as follows:
\begin{equation}
r(\alpha) \leq \frac{\alpha}{\lceil \frac 1 {\alpha} \rceil}\leq
\frac {\frac{1}{h-1}}{\lceil \frac 1 h \rceil} \leq \frac{h}{h-1}
\end{equation}
This means that for smaller values of $\alpha$ the we get an almost tight
description of $\pof(\alpha)$.

\subsection{Shared item set}
\label{sec:shared}
In this section we assume that the agents $A$ and $B$
share a joint set of items $w_1, w_2, \ldots, w_n$
with $w_i \leq 1$ and $\sum_{i=1}^n w_i >1$.
Of course, each item can be assigned to at most one of the agents.

As already observed, this scenario with a shared item set is closely related to {Fair Division},
more precisely to the division of indivisible goods \cite{bib:bt1996,bib:kla10}.
However, differently from Fair Division, we consider a capacity, i.e.\
a condition that not all given items should be partitioned between $A$ and $B$,
but only a subset which can be freely chosen as long
as its total weight does not exceed the capacity.

Note that, unlike the separate items case, here once the subset
$S \subset \{w_1, w_2, \ldots\}$ of items is  chosen,
each  bipartition $(x^A, x^B)$ of $S$ corresponds to a feasible solution of our problem.
As a consequence,  there can be exponentially many distinct solutions corresponding to
the same subset $S$ and therefore returning the same global utility value $U(S)$.
In the sequel, when needed, we  specify which partition of a certain subset of items is considered as a solution.

Note also that the shared items case is a special case of the symmetric problem considered in Section \ref{sec:symmetry}, so Theorem~\ref{th:shared_pf_sameamount}
and Corollary~\ref{thm:pfsystemoptimum} hold and imply
$\pof_{PF}=0$ whenever a proportional fair solution exists.
%
%
Moreover, concerning Kalai-Smorodinsky fairness, in the shared items case
we trivially  have $x_{KS}=x_{MM}$ since $\bestA=\bestB$.
Therefore, in the following we  refer only to  maximin fair solutions.

We first present some lower bounds on $\pof$ through Examples \ref{ex:ssppoflargeshared} and \ref{ex:shared_pof_small_alpha} and then provide the matching upper bounds in Theorem \ref{th:ssppofshared}.

\smallskip
For the case with no bound on the weights ($\alpha=1$),
we can use the item set of Example~\ref{ex:ssppof} as a common ground set.
It is easy to show that a maximin fair solution has a value $\z(x_{MM}) = 3\eps$ and hence that $\pof_{MM} \to 1$ when $\eps \to 0$.
For $\alpha<1$  we give the following two examples to derive lower bounds on $\pof_{MM}(\alpha)$.

\begin{example}\label{ex:ssppoflargeshared}
Consider an instance of FSSP with shared items with $\alpha \in [2/3, 1)$.
For a small constant $\eps>0$ let the items weights be as follows.
$$
\begin{array}{lccccc}
\hline
\mbox{item}  & w_1  &  w_2 & w_3  & w_4  \tabularnewline
\hline
\mbox{weight}  &\alpha & 1-\alpha+\eps & 1-\alpha & \eps\tabularnewline
\hline
\end{array}
$$
An optimal solution $\opt$ with $\z^*=1$ consists of $\A(\opt)=\alpha$ and $\B(\opt)=1-\alpha$.
The only PO solution $x$ improving  $B$'s utility cannot select $w_1$ yielding $\A(x) = \B(x) = 1-\alpha+\eps$.
So,
$$\pof_{MM}(\alpha) \geq  \frac{1-(2-2\alpha+2\eps)}{1} \to 2\alpha -1.$$
\end{example}

\begin{example}\label{ex:shared_pof_small_alpha}
Consider an instance of FSSP with shared items with $\frac{1}{2h+1} \leq \alpha<\frac{1}{2h-1}$ for some integer $h\geq 1$ and the following items set.
$$
\begin{array}{lcc}
\hline
\mbox{item}  & w_1 = w_2 =\ldots =w_{2h+1} & \ \ w_{2h+2}=w_{2h+3} \tabularnewline
\hline
\mbox{weight}  & 1/(2h+1) & \eps \tabularnewline
\hline
\end{array}
$$
Clearly, an optimal solution $\opt$ with $\z^*=1$ consists of
$\A(\opt)=(h+1)/(2h+1)$ and $\B(\opt)=h /(2h+1)$,
while the maximin fair solution $x_{MM}$ is such that
$\A(x_{MM})=\B(x_{MM})=h/(2h+1) + \eps$.
This yields:
\begin{equation}\label{eq:lowerbound_h_shared}
\pof(\alpha) \geq \pof\left(\frac{1}{2h+1}\right) \geq
 \frac{1-\frac{2h}{2h+1} -2\eps}{1} \to \frac{1}{2h+1}
\end{equation}

Note that since $\alpha$ is an upper bound on the largest item weight of the instance, we may express the lower bound on the Price of Fairness in terms of $\alpha$. For any $\alpha \in [\frac 1 {2h+1 }, \frac 1 {2h-1})$  the lower bound is
\begin{equation}\label{eq:lowerbound_alpha_shared}
\pof(\alpha) \geq \pof\left(\frac 1 {2h+1 }\right) \geq
\frac 1 {2h+1 }\geq \frac {1} {2 \lceil \frac 1 {2\alpha}\rceil+1}.
\end{equation}
%
\end{example}
Observe that for $h=1$ Example~\ref{ex:shared_pof_small_alpha}
yields a lower bound of $1/3$ which matches the lower bound of
Example~\ref{ex:ssppoflargeshared} for $\alpha=2/3$.

The next theorem provides upper bounds on $\pof_{MM}(\alpha)$ (that match the lower bounds of Examples~\ref{ex:ssppoflargeshared} and \ref{ex:shared_pof_small_alpha} for $\alpha \geq 1/3$) for the shared items sets case, as it is shown in Figure~\ref{fig:poftrendsshared}.
\begin{theorem}\label{th:ssppofshared}
FSSP with shared item set and an upper bound $\alpha$ on the maximum item weight
has the following Price of Fairness for maximin fair solutions:
 \begin{align}
\pof_{MM}( \alpha)& =  2\alpha -1   & \mbox{ for } 2/3 < \alpha \leq 1 \label{eq:ssps1}\\
\pof_{MM}( \alpha) & =   1/3 & \mbox{ for } 1/3 < \alpha \leq 2/3\label{eq:ssps2}\\
\frac {1} {2 \lceil \frac 1 {2\alpha}\rceil+1} \ \leq \
\pof_{MM}( \alpha) & \leq    \alpha  & \mbox{ for } 0< \alpha \leq 1/3 \label{eq:ssps3}
\end{align}
\end{theorem}
\begin{proof}
The tight lower bounds for $\alpha \geq 1/3$ were shown by
Examples~\ref{ex:ssppoflargeshared} and \ref{ex:shared_pof_small_alpha}
(for $h=1$).
The latter example also provides the lower bound of \eqref{eq:ssps3}.

The upper bound in \eqref{eq:ssps3} for $\alpha \leq 1/3$ follows immediately from Lemma~\ref{lem:lessthanalpha}.

\smallskip
For $\alpha \in (1/3, 1]$ we proceed as follows.
Let $\z^*>\z(x_{MM})$  for a fair solution $x_{MM}$ such that
$\A(x_{MM})\geq \B(x_{MM})$.
By definition of $MM$ and Pareto efficiency we must have
\begin{equation}\label{eq:mmopt}
\max\{\A(\opt), \B(\opt)\} > \A(x_{MM}) \geq \B(x_{MM}) > \min\{\A(\opt), \B(\opt)\}.
\end{equation}
Now we  consider two cases depending
on the weight $\bar{w} \leq \alpha$ of the largest item contained in a system optimal solution $\opt$.
\begin{itemize}
\item {\em Case 1}: $\bar{w} \geq 1/2$.
Among the different system optima consider the one where $\A(\opt)=\bar{w}$,
while $\B(\opt)$ corresponds to the  weight of some other subset of items.
Clearly, $\B(\opt) \leq 1-\bar{w}\leq \bar{w}=\A(\opt)$ and
neither $\A(x_{MM})$ nor $\B(x_{MM})$ contains $\bar{w}$.
We have that $\A(x_{MM}) \geq  1-\bar{w}$ since otherwise,
i.e.\ if $\A(x_{MM}) <  1-\bar{w} \leq 1/2$, we could add $\bar{w}$ to
$\A(x_{MM})$ which then exceeds $\A(\opt)$.
Since also $\B(x_{MM}) > \B(\opt)$ this would constitute a solution with
better total value than $\opt$.
By a similar argument also $\B(x_{MM}) \geq 1- \bar{w}$.

Hence, $\z(x_{MM}) \geq 2-2\bar{w}$ and we get the upper bound
\begin{equation*}
\pof_{MM}(\alpha)\leq  \frac{\bar{w} + \B(\opt) -(2-2\bar{w})}{\bar{w} + \B(\opt)} =
\frac{3\bar{w} -2 + \B(\opt)}{\bar{w} + \B(\opt)}
\end{equation*}
which is increasing in $\B(\opt)$ for all $\bar{w} \leq 1$.
Thus, by plugging in the largest possible $\B(\opt)$ value, that is $\B(\opt)=1-\bar{w}$,
we obtain for $\bar{w} \leq \alpha$
\begin{equation}\label{eq:ssppofshared}
\pof_{MM}(\alpha) \leq 2\bar{w} -1 \leq 2 \alpha -1.
\end{equation}  

\item {\em Case 2}: $\bar{w} < 1/2$.
Among the different system optima consider the one built with an LPT like procedure for $P2||C_{\max}$ (see for instance \cite{gllr79}):
The items in $\opt$ are sorted in decreasing order and assigned iteratively to the
agent with current lower total weight. 
Let $\A(\opt)$ and $\B(\opt)$ indicate the values for the two agents in this solution.
Clearly, in general, it is not known which of the two values is larger.

If $\max\{\A(\opt), \B(\opt)\}= \bar{w}$ then following (\ref{eq:mmopt})
any solution $\B(x_{MM})$ with
$1/2 > \max\{\A(\opt), \B(\opt)\} > \B(x_{MM}) > \min\{\A(\opt), \B(\opt)\}$
could be used to replace and improve  $\min\{\A(\opt), \B(\opt)\}$ in $\opt$.
Hence, it must be $\max\{\A(\opt), \B(\opt)\} > \bar{w}$.
This means that according to the LPT logic, at least one additional item was
added to the agent receiving $\bar{w}$,
which can happen only after the other agent weight has exceeded $\bar{w}$.
Therefore, $\min\{\A(\opt), \B(\opt)\} \geq \bar{w}$.

By LPT we also have
$|\A(\opt)  - \B(\opt) | \leq \bar{w}$.
It follows with (\ref{eq:mmopt}) that
$$\max\{\A(\opt), \B(\opt)\} \leq \min\{\A(\opt), \B(\opt)\} + \bar{w}
 \leq 2 \min\{\A(\opt), \B(\opt)\} < 2\, \B(x_{MM}).$$
Thus, we  have
 $$\z^*= \max\{\A(\opt), \B(\opt)\} + \min\{\A(\opt), \B(\opt)\}
< 3\, \B(x_{MM}) \leq 3/2\: z(x_{MM}).$$

It follows immediately that, for $\bar{w}<1/2$, independently from $\alpha>\bar{w}$ ,
 \begin{equation}\label{eq:pofonethird}
 \pof_{MM}(\alpha)\leq 1/3.
 \end{equation}
\end{itemize}
While it is clear that for $\alpha < 1/2$ only Case~2 is feasible,
for an instance with $\alpha \geq 1/2$ either of the two cases may occur.
Hence, we can only state an upper bound as a maximum of the two:
\begin{equation}
\pof_{MM}(\alpha) \leq \max\left\{2\alpha -1, \frac 1 3\right\}
\end{equation}
which easily yields relations
\eqref{eq:ssps1} and \eqref{eq:ssps2}.
\end{proof}
\begin{figure}[b!]
    \begin{center}
        \subfigure[Separate items]{%
            \label{fig:poftrendsseparate}
            \includegraphics[width=0.45\textwidth]{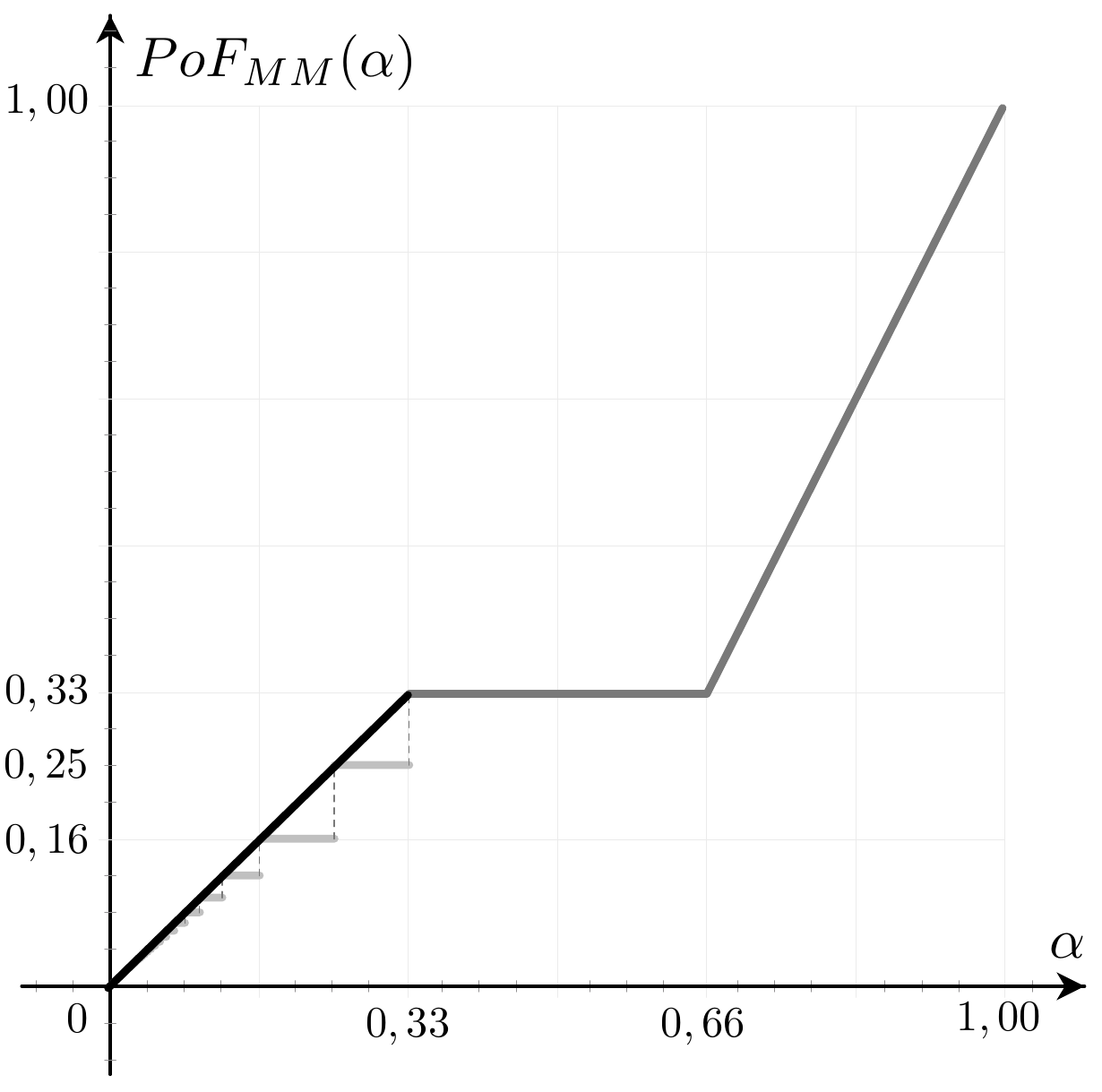}
        }
        \hfill
        \subfigure[Shared items]{%
           \label{fig:poftrendsshared}
           \includegraphics[width=0.45\textwidth]{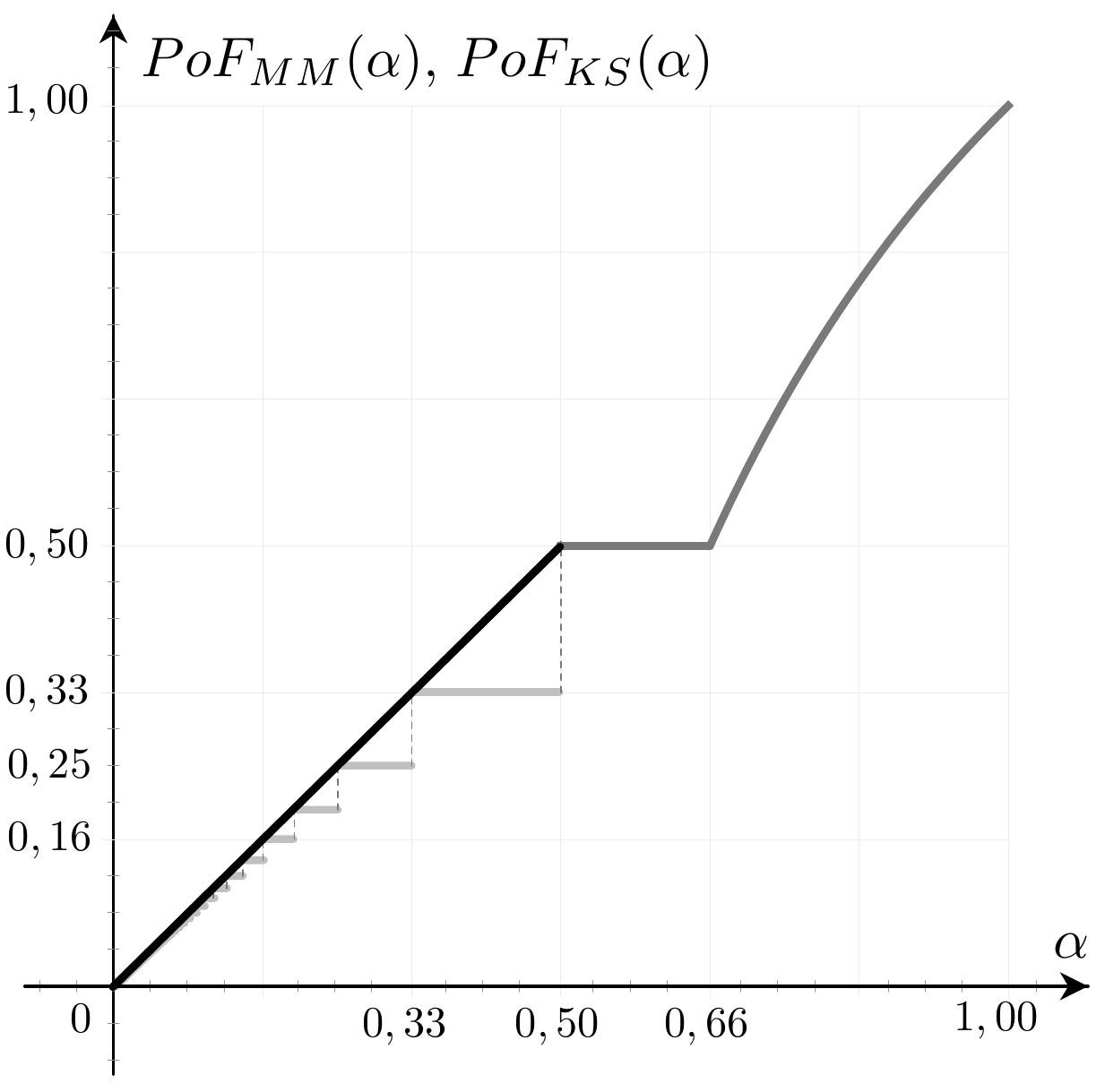}
        }
    \caption{%
        Price of Fairness functions $\pof_{MM}(\alpha)$ and $\pof_{KS}(\alpha)$ in the separate items sets case (cf. Theorems~\ref{th:ssppof} and \ref{th:ssppofks}) and $\pof_{MM}(\alpha)$ in the shared items set case (cf. Theorem~\ref{th:ssppofshared}). The corresponding lower and upper bound values are plotted, when needed.
    }
    \label{fig:poftrends}
    \end{center}
\end{figure}
For $\alpha < 1/3$, where the bound of Theorem~\ref{th:ssppofshared} is not tight,
 we can bound (as in the case of separate item sets)
the ratio $r(\alpha)$ between upper and lower bound in \eqref{eq:lowerbound_h_shared}
on $\pof_{MM}(\alpha)$ as follows:
\begin{equation}
r (\alpha ) = \alpha \cdot (2h+1) < \frac{2h+1}{2h-1}
\end{equation}
Again, this shows that for smaller values of $\alpha$
an almost tight description of $\pof_{MM}(\alpha)$ is derived.
The largest gap arises for $\alpha= 1/3 -\eps$
where $1/5 \leq \pof_{MM}(\alpha) < 1/3$.

\section{Conclusions}
\label{sec:conc}

In this paper we introduced a general allocation function to assign utilities to a set of agents.
The focus of our attention is directed on {\em fair allocations} which give a reasonable amount of utility to each agent.
A number of fairly general results holding for any multi-agent problem were derived for three different notions of fairness, namely
maximin, Kalai-Smorodinsky and proportional fairness.
In particular, we showed that for a large and meaningful class of problems
proportional fair solutions are system optimal and equitable, that is each agent receives the same utility as every other agent.

In the main part of the paper we considered a bounded resource allocation problem
which can be seen as a two-agent version of the subset sum problem
and thus is referred to as {\em Fair Subset Sum Problem} (FSSP).
We are interested in evaluating the loss of efficiency
incurred by a fair solution compared to a system optimal solution which maximizes the sum of agents utilities.
In particular, we presented several lower and upper bounds on the Price of Fairness for different versions of the  problem.

As discussed for the three notions of fairness considered in this paper,
it is in general hard to compute a fair solution,
so it would be desirable to introduce a solution concept permitting a polynomial time algorithm,
or even a simple heuristic allocation rule,
fulfilling some fairness criterion and still guaranteeing an adequate level of efficiency
(i.e.\ a certain upper bound on the Price of Fairness).

 Concerning  FSSP, it is easy to show that it is binary NP-hard to recognize fair solutions (for  all three fairness concepts). In fact,  if  all item weights and the capacity $c$ are integers, it is possible  to design  dynamic programming algorithms running in pseudopolynomial time to find all PO solutions in the separate and shared items cases. The algorithms are briefly sketched hereafter.
For separate items, we may define two dynamic arrays $d^A[w]$, $d^B[w]$, $w=0,1,\ldots, c$,
with binary entries, where e.g.\ $d^A[w]=1$ indicates that a solution with weight $w$ exists for agent $A$.
The entries can be easily computed by iteratively considering each item, say $a_i$,
and setting $d^A[w+a_i]=1$ if $d^A[w]=1$.
Finally, we go through both arrays in opposite directions and identify all
Pareto efficient combinations of weights $w, v$ with $w+v\leq c$ and $d^A[w]=1$ and $d^B[v]=1$. This takes $O(nc)$ time.
For shared items, a two-dimensional array is required, where $d[w,v]=1$
if a solution with weight $w$ for $A$ and $v$ for $B$ exists.
It is updated for each item $w_i$ by observing that each entry with $d[w,v]=1$
implies that both $d[w+w_i,v]=1$ and $d[w,v+w_i]=1.$
Thus, all reachable solutions can be determined in $O(nc^2)$ time.
More details, e.g.\ about storing the set of items for each entry, can be found
in~\cite[Sec.~2.3]{kpp04}.

Finally, a natural generalization of the FSSP, with significant applications in several real-world scenarios such as Project Management and Portfolio Optimization, would consider  a different utility function associated to  profits, thus defining a multi-agent (fair) knapsack problem.

\subsection*{Acknowledgements}

  Gaia Nicosia and Andrea Pacifici have been partially supported by
  Italian MIUR projects PRIN-COFIN n.~2012JXB3YF 004 and
  n.~2012C4E3KT 001.\\
  Ulrich Pferschy was supported by the Austrian
  Science Fund (FWF): 
  P 23829-N13.


\newpage

\section{Appendix}

\noindent
\textbf{Theorem \ref{thm:PFequal}}
\emph{If two proportional fair solutions $x_{PF}$ and $y_{PF}$  exist, then $u_j(x_{PF})=u_j(y_{PF})$ for all $j=1,\ldots, k$.}

\medskip
\noindent
\begin{proof}
Let $x_{PF}$ and $y_{PF}$ be two proportional fair solutions.
By definition of proportional fairness and using equation \eqref{eq:def_propfair}
for both  $x_{PF}$  and $y_{PF}$,  we obtain
$\sum_{j=1}^k \frac{u_j(y_{PF})}{u_j(x_{PF})} \leq k$ and $\sum_{j=1}^k \frac{u_j(x_{PF})}{u_j(y_{PF})} \leq k$.
Let $\phi_j= \frac{u_j(x_{PF})}{u_j(y_{PF})}$ for $j=1,\ldots k$, clearly $\phi_j\geq 0$. Then the two above inequalities can be rewritten as:
$\sum_{j=1}^k \frac{1}{\phi_j} \leq k$ and $\sum_{j=1}^k \phi_j \leq k$.
By summing up these last two inequalities we get that $\sum_{j=1}^k (\frac{1}{\phi_j} +\phi_j ) \leq 2k$. Moreover,  $\frac{1}{\phi_j} +\phi_j  \geq 2$ for any $\phi_j\,$.
Hence, the only possible way to satisfy $\sum_{j=1}^k (\frac{1}{\phi_j} +\phi_j  )\leq 2k$ is
$\frac{1}{\phi_j} +\phi_j  = 2$, which implies $\phi_j = 1$, for all $j=1,\ldots, k$.
\end{proof}

\medskip
\noindent
\textbf{Theorem \ref{th:nash}}
\emph{If a proportional fair solution $x_{PF}$ exist, then it maximizes the product of agents utilities, i.e.\
$$\prod_{j=1}^k u_j(x_{PF}) \geq \prod_{j=1}^k u_j(x) \quad \forall x\, \in X.$$}

\medskip
\noindent
\begin{proof}
Let $x_{PF}$ be the proportional fair solution and $y \in X$ any feasible solution. Let $\phi_j = \frac{u_j(y) }{u_j(x_{PF})}$.
By \eqref{eq:def_propfair}, recalling that the geometric mean is not
larger than the arithmetic mean, we have
$$\left(\prod_{j=1}^k \phi_j\right)^{\frac 1 k } \leq \frac 1 k  \sum_{j=1}^k \phi_j\leq 1.$$
As a consequence $\prod_{j=1}^k u_j(y) \leq \prod_{j=1}^k u_j(x_{PF})$ and  the thesis follows.
\end{proof}
%





\begin{thebibliography}{99}




\bibitem{audo10}
Aumann Y., Y. Dombb (2010).
The efficiency of fair division with connected pieces,
 Proceedings of WINE 2010,
\emph{Springer Lecture Notes in Computer Science}, 6484, 26-37.

\bibitem{bib:bft2011}
Bertsimas D., V. Farias, N. Trichakis (2011).
The price of fairness, \emph{Operations Research}, 59 (1), 17--31.

\bibitem{bib:bft2012}
Bertsimas D., V. Farias, N. Trichakis (2012).
On the efficiency-fairness trade-off,
\emph{Management Science},  58(12),  2234--2250.



\bibitem{bib:bt1996}
Brams S.J.,  A.D. Taylor (1996).
\emph{Fair Division: From cake-cutting to dispute resolution},
Cambridge University Press.

\bibitem{bib:bw2002}
Butler, M., H.P. Williams (2002).
Fairness versus efficiency in charging for the use of common facilities,
\emph{Journal of Operational Research Society},  53(12), 1324--1329.

\bibitem{bib:ckkk2009}
Caragiannis I., C. Kaklamanis, P. Kanellopoulos, M. Kyropoulou (2012).
The efficiency of fair division,  
\emph{Theory of Computing Systems}, 50(4), 589--610, 2012.
See also: Proceedings of WINE 2009,
\emph{Springer Lecture Notes in Computer Science}, 5929, 475--482.

\bibitem{hoch97}
Coffman Jr., E.G., M.R. Garey, D.S. Johnson (1997).
Approximation algorithms for bin packing: a survey,
in: D. Hochbaum (Ed.), \emph{Approximation Algorithms for NP-hard Problems}, 
PWS Publishing Co.

\bibitem{dnps13}
Darmann A., G. Nicosia, U. Pferschy, J. Schauer (2014).
The subset sum game,
\emph{European Journal of Operational Research}, 233(3), 539--549.


\bibitem{bib:budgetgames}
Drees M., S. Riechers, A. Skopalik (2014).
Budget-restricted utility games with ordered strategic decisions,
Proceedings of SAGT 2014,
\emph{Springer Lecture Notes in Computer Science}, 8768, 110--121.

\bibitem{bib:frf2015} Fritzsche R., P. Rost,  G.P. Fettweis (2015).
Robust rate adaptation and proportional fair scheduling with imperfect CSI,
\emph{IEEE Transactions on Wireless Communications}, 14(8), 4417 - 4427.


\bibitem{fy04}
Fujimoto M., T. Yamada (2006).
An exact algorithm for the knapsack sharing problem with common items,
\emph{European Journal of Operational Research}, 171(2), 693--707.

\bibitem{bib:gzhkss11} Ghodsi A., M. Zaharia, B. Hindman, A. Konwinski, S. Shenker, and I. Stoica (2011). Dominant Resource Fairness: Fair Allocation of Multiple Resource Types, Proceedings of the 8th USENIX Conference on
Networked Systems Design and Implementation (NSDI), 24--37.

\bibitem{bib:gkw2010}
Goel G., C. Karande, L. Wang (2010).
Single-parameter combinatorial auctions with partially public valuations,
Proceeding of SAGT 2010,
\emph{Springer Lecture Notes in Computer Science}, 6386, 234--245.

\bibitem{gllr79}
Graham R.L., E.L. Lawler, J.K. Lenstra, A.H.G. Rinnooy Kan (1979),
Optimization and approximation in deterministic sequencing and scheduling: a survey,
in: P.L. Hammer et al. (Eds.),
\emph{Annals of Discrete Mathematics}, 5, 287--326, Elsevier.

\bibitem{hifi2005}
Hifi M., H. M'Hallab, S. Sadfi (2005).
An exact algorithm for the knapsack sharing problem,
\emph{Computers and Operations Research}, 32(5), 1311--1324.

\bibitem{bib:ks1975}
Kalai E., M. Smorodinsky (1975).
Other solutions to Nash bargaining problem,
\emph{Econometrica}, 43, 513--518.

\bibitem{bib:km2015}
Karsu \"O., A. Morton (2015),
Inequity averse optimization in operational research,
\emph{European Journal of Operational Research}, 245(2), 343--359.

\bibitem{kpp04}
Kellerer H., U. Pferschy, D. Pisinger (2004).
\emph{Knapsack Problems},
Springer.

\bibitem{bib:kmt98}
Kelly F.P., A.K. Maulloo and D.K.H. Tan (1998).
Rate control in communication networks: shadow prices, proportional fairness and stability,
\emph{Journal of the Operational Research Society}, 49, 237--252.

\bibitem{bib:kla10}
Klamler, C. (2010).
Fair Division,
in: Kilgour D.M. and C. Eden (Eds.),
\emph{Handbook of Group Decision and Negotiation}, Springer, 183--202.

\bibitem{kyo12}
K\"oppen M., K. Yoshida, K. Ohnishi, M. Tsuru (2012).
Meta-heuristic approach to proportional fairness,
\emph{Evolutionary Intelligence}, 5(4), 231--244.

\bibitem{bib:k2009}
Kozanidis G. (2009).
Solving the linear multiple choice knapsack problem with two objectives:
profit and equity,
\emph{Computational Optimization and Applications}, 43(2), 261--294.



\bibitem{bib:sagt2015}
Nicosia G., A. Pacifici, U. Pferschy (2015).
Brief announcement: On the fair subset sum problem,
Proceedings of SAGT 2015,
\emph{Springer Lecture Notes in Computer Science}, 9347, 309--311.

\bibitem{bib:PoA}
Nisan N., T. Roughgarden, E. Tardos,  V.V. Vazirani (2007).
\emph{Algorithmic Game Theory},
Cambridge University Press.

\bibitem{bib:pps15}
Parkes D.C., A.D. Procaccia,  N. Shah (2015).
Beyond dominant resource fairness: extensions, limitations, and indivisibilities.
\emph{ACM Transactions on Economics and Computation}, 3(1), Article No. 3.

%

\bibitem{bib:raw71}
Rawls J. (1971).
\emph{A Theory of Justice},
Harvard University Press.


\bibitem{zhsh15}
Zhang C., J.A. Shah (2015).
On fairness in decision-making under uncertainty: definitions, computation, and comparison.
\emph{Proceedings of the 29th AAAI Conference on Artificial Intelligence},
3642--3648.

\end{thebibliography}

\end{document}